\documentclass{article}

\makeatother
\usepackage{multicol,placeins}
\usepackage[preprint]{spconf}
\usepackage{flushend}
\usepackage{graphicx}
\usepackage{tikz}
\usepackage{color}
\usepackage{amsmath, amsfonts, amssymb, mathrsfs}
\usepackage[amsmath,thmmarks]{ntheorem}
\usepackage{subfigure}
\usepackage{epstopdf}
\usepackage{pgfplots}
\usepackage{arydshln}
\makeatletter
\renewcommand*\env@matrix[1][*\c@MaxMatrixCols c]{%
  \hskip -\arraycolsep
  \let\@ifnextchar\new@ifnextchar
  \array{#1}} 

\onecolumn
\theoremstyle{plain}
\theoremheaderfont{\itshape\bfseries}
\theorembodyfont{\normalfont\itshape}
\theoremseparator{:}
\theoremsymbol{}
\newtheorem{theorem}{Theorem}

\theoremstyle{nonumberplain}

\theoremstyle{nonumberplain}
\theoremheaderfont{\itshape}
\theorembodyfont{\normalfont}
\theoremseparator{:}
\theoremsymbol{}
\newtheorem{remark}{Remark}

\theoremstyle{plain}
\theoremheaderfont{\itshape\bfseries}
\theorembodyfont{\normalfont}
\theoremseparator{:}

\theoremstyle{nonumberplain}
\theoremsymbol{\rule{1ex}{1ex}}

\newtheorem{proof}{Proof}
\newlength\fheight
\newlength\fwidth
\def\Hx{ \widehat{x} }
\def\Hy{ \widehat{y} }

\def\I{ \mathrm{i} }							
\def\E{ \mathrm{e} }							
\def\T{ \mathrm{T} }							
\def\RN{ \mathbb{R} }							
\def\CN{ \mathbb{C} }							
\def\TT{ \mathbb{T} }                             	
\def\uC{ \mathbb{T} }

\DeclareMathOperator*{\supp}{supp}			

\title{Fast Compressive Phase Retrieval from Fourier Measurements}
\name{
  \c{C}a\u{g}kan~Yapar,
  Volker~Pohl,  
  Holger~Boche\thanks{This work was partly supported by the German Research Foundation (DFG) under Grant PO~1347/2-1.}
}
\address{
  Lehrstuhl f{\"u}r Theoretische Informationstechnik\\
  Technische Universit{\"a}t M{\"u}nchen, 80290 M{\"u}nchen, Germany\\
  \{cagkan.yapar, volker.pohl, boche\}@tum.de
}

\begin{document}
%
\maketitle
%
\begin{abstract}
This paper considers the problem of recovering a $k$-sparse, $N$-dimensional complex signal from Fourier magnitude measurements. It proposes a Fourier optics setup such that signal recovery up to a global phase factor is possible with very high probability whenever $M \gtrsim 4k\log_2(N/k)$ random Fourier intensity measurements are available. The proposed algorithm is comprised of two stages: An algebraic phase retrieval stage and a compressive sensing step subsequent to it.
Simulation results are provided to demonstrate the applicability of the algorithm for noiseless and noisy scenarios. 
\end{abstract}

\section{Introduction}
\label{sec:intro}

In many applications involving linear signal measurement processes, the measurement results are magnitude-only or solely unreliable phase information is available. 
Phase retrieval addresses this problem by striving to recover the signal
exclusively from the absolute values of the linear measurements. Fourier optics is one of the application areas where the phase retrieval problem is commonly faced. An exemplary setting is shown in Fig.~\ref{fig:Optics} (the mask belongs to the recovery setup, assume it to be nonexistent for the moment). The object of interest is illuminated by a light or x-ray source. As a result, a diffraction pattern $x[n]$ is produced, where $n$ denotes the discrete spatial coordinate. Subsequently, this diffraction pattern $x[n]$ is transformed by the lens into the Fourier domain. Unable to measure the phase, one can  only acquire the intensity measurements $|\widehat{x}[\omega]|^{2}$ of the Fourier transform $\widehat{x}[\omega]$. The phase retrieval problem is now to reconstruct the diffraction pattern $x[n]$ from the intensity measurements $|\widehat{x}[\omega]|^{2}$. In this paper, we are interested in the case where $x\in \CN^{N}$ is known to be $k$-sparse.

Non-sparse phase retrieval has been a very active research area since the seminal work of Balan et al.~\cite{Balan_RecWithoutPhase_06}. It is proven in~\cite{Conca_Algebraic13} that $4N-4$ measurements are sufficient, and in~\cite{Heinosaari_QuantumTom_13} that $4N-o(N)$ measurements are necessary for perfect recovery up to a global phase factor. Minimal deterministic constructions yielding injectivity with $4N-4$ measurement vectors are provided in ~\cite{Bodmann_StablePR2014,Fickus_VeryFewMeasurements}. Also, explicit deterministic measurement ensembles ensuring injectivity for almost every signal in $\CN^N$ are proposed in~\cite{PYB_SampTa14,phasecode}.

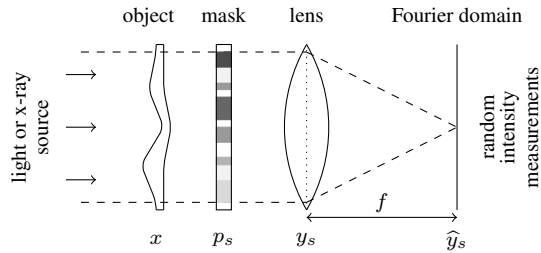
\begin{figure}[ht]
\begin{center}
\begin{tikzpicture}
	\draw (-3.8,0) node[rotate = 90] {\footnotesize light or x-ray};
	\draw (-3.5,0) node[rotate = 90] {\footnotesize source};
	\draw[->] (-3.2,-0.7) -- (-2.8,-0.7);
	\draw[->] (-3.2,0) -- (-2.8,0);
	\draw[->] (-3.2,0.7) -- (-2.8,0.7);
	\draw[dashed] (-3,1) -- (0,1);
	\draw[dashed] (-3,-1) -- (0,-1);
	\draw[dashed] (0,1) -- (2,0);
	\draw[dashed] (0,-1) -- (2,0);
	\draw (-2.1,1.46) node {\footnotesize object};
	\draw[rounded corners = 1mm] (-2, -1.1) -- (-2,-1) -- (-2.2,-0.5) -- (-1.9,0) -- (-2.1,0.5) -- (-2,1) -- (-2,1.1);
	\draw[rounded corners = 1mm] (-1.9, -1.1) -- (-1.9,-1) -- (-1.9,-0.5) -- (-1.8,0) -- (-1.9,0.5) -- (-1.9,1) -- (-1.9,1.1);
	\draw (-1.9,-1.1) -- (-2,-1.1);
	\draw (-1.9,1.1) -- (-2,1.1);
	\draw  (-2,-1.5) node {\footnotesize $x$};
	\draw (-1.1,1.5) node {\footnotesize mask};
	\filldraw[fill = black!15!white, draw = black!15!white] (-1.2,-1) rectangle (-1,-0.7);
	\filldraw[fill = black!05!white, draw = black!05!white] (-1.2,-0.7) rectangle (-1,-0.4);
	\filldraw[fill = black!30!white, draw = black!30!white] (-1.2,-0.5) rectangle (-1,-0.4);
	\filldraw[fill = black!40!white, draw = black!40!white] (-1.2,-0.2) rectangle (-1,-0.0);
	\filldraw[fill = black!60!white, draw = black!60!white] (-1.2,0.1) rectangle (-1,0.4);
	\filldraw[fill = black!40!white, draw = black!40!white] (-1.2,0.5) rectangle (-1,0.6);
	\filldraw[fill = black!05!white, draw = black!05!white] (-1.2,0.6) rectangle (-1,0.8);
	\filldraw[fill = black!70!white, draw = black!70!white] (-1.2,0.8) rectangle (-1,1.0);
	\draw (-1.2,-1.1) rectangle (-1,1.1);
	\draw  (-1.1,-1.5) node {\footnotesize $p_{s}$};
	\draw  (0,1.5) node {\footnotesize lens};
	\draw (0,1.1) arc (150:210:2.2);
	\draw (0,-1.1) arc (-30:30:2.2);
	\draw[dotted] (0,-1.1) -- (0,1.1);
	\draw  (0,-1.5) node {\footnotesize $y_{s}$};
	\draw  (2,1.5) node {\footnotesize Fourier domain};
	\draw (2,1.1) -- (2,-1.1);
	\draw  (2,-1.5) node {\footnotesize $\Hy_{s}$};
	\draw[<->] (0,-1.2) -- (2,-1.2);
	\draw (1,-1) node {\footnotesize $f$};
	\draw (2.4,0) node[rotate = 90] {\footnotesize random};
	\draw (2.7,0) node[rotate = 90] {\footnotesize intensity};
	\draw (3.0,0) node[rotate = 90] {\footnotesize measurements};
\end{tikzpicture}
\end{center}
\caption{A setup for phase retrieval problem which can often be found in optical applications.}
\label{fig:Optics}
\end{figure}



Compressive phase retrieval of sparse signals attracted some interest in recent years.
It was shown in \cite{LiVoro_2012} that $8 k - 2$ generic intensity are sufficient for recovery whereas they needed $\mathcal{O}(k \ln N )$ measurements for stable recovery via convex programming.
The theoretical lower bound for the number of sufficient measurements required for a $k$-sparse signal was improved in \cite{akcakaya} to $M = 4k-2$. However, to the best of our knowledge there is no algorithm approaching this bound presently. Using {\em PhaseLift} \cite{Candes_PhaseLift},  Ohlsson et al.~\cite{CPRL} proposed an recovery algorithm from $\mathcal{O}(k^2 \log N)$ measurements. However, this technique is based on semidefinite programming and suffers from high computational complexity. In~\cite{prGAMP}, a technique relying on generalized approximate message passing is presented. While the simulation results in this work demonstrate some advantages in terms of the number of the required measurements and computational complexity,
no theoretical recovery guarantee was derived.

The characteristics of the measurement vectors play a key role in the practical applicability of the algorithms. None of the previously mentioned works focus on measurement sets that could model a Fourier optics system as in Fig.~\ref{fig:Optics} (see, e.g., \cite{BandChen_II14,Candes_CDP13,CandesEldar_PhaseRetrieval,Falldorf_SLM10,Gross_2014,Xiao_DistortedObject05,Zhang_ApatureMod07}).  
In fact, the first paper about compressive phase retrieval~\cite{Moravec} was addressing the very problem that we are trying to solve in the present paper, i.e., the recovery problem of a $k$-sparse complex signal $x \in \CN^N$ from Fourier intensity measurements  $|\widehat{x}[\omega]|^{2}$. To our knowledge, the only work after~\cite{Moravec} that directly addressed this problem is the recent paper by Pedarsani et al.~\cite{phasecode}. Based on a sparse graph codes framework, this paper provides a low complexity algorithm that achieves perfect reconstruction with very high probability using $14k$ measurements. 

The present paper proposes a two step procedure to recover almost every $x \in \CN^N$ by random Fourier intensity measurements using 4 masks (see Fig.~\ref{fig:Optics}). First, we recover the phases of our measurements up to a global phase  using the algorithm proposed in~\cite{PYB_SampTa14}. Afterward, the sparse signal is reconstructed using the standard compressed sensing approach, i.e., the $\ell_{1}$-minimization technique. We provide numerical simulations to show the success rates of the algorithm and its behavior under additive measurement noise.

\section{Signal Model and Notations}
\label{sec:SignalModel}

\paragraph*{Notations}
We consider signals in the $N$-dimensional complex Euclidean vector space $\CN^{N}$.
These signals are written as $x = (x[1] , x[2], \dots, x[N])^{\T}$.
The inner product in $\CN^{N}$ is $\left\langle x,y \right\rangle_{\CN^{N}} = \sum^{N}_{n=1} x[n]\, \overline{y[n]} = y^{*}x$ where the bar denotes the complex conjugate, and $y^{*}$ is the conjugate transpose of $y$. The norm, induced by the inner product is denoted by $\|x\| = \sqrt{\left\langle x,x\right\rangle}$, whereas $\|x\|_{\ell_{1}} := \sum^{N}_{n=1}|x[n]|$ stands for the $\ell_{1}$ norm.
The \emph{unitary discrete Fourier transform (DFT)} of $x\in\CN^{N}$ is given by
$\Hx = F x$ where $F$ denotes the DFT matrix with entries
\begin{equation*}
	[F]_{m,n}=\tfrac{1}{\sqrt{N}}\E^{-i2\pi(m-1)(n-1)/N}\;,\quad
	m,n=1,\dots,N
\end{equation*}
in its $m$th row and $n$th column.

A vector $x \in \CN^{N}$ is called $k$-sparse if $\|x\|_{0} := |\supp(x)| \leq k$ where $\supp(x) = \{n : x[n] \neq 0\}$, 
i.e., if the support length is at most equal to $k$. The set of all $k$-sparse vectors in $\CN^{N}$ is denoted by
\begin{equation*}
	\Sigma^{N}_{k} = \{ x \in \CN^{N} : \|x\|_{0} \leq k\}\;.
\end{equation*}
We write $x\odot y$ for the point-wise product of two vectors $x,y \in \CN^{N}$, i.e., $(x\odot y)[n] = x[n] y[n]$ for all $n = 1,\dots,N$, and 
$\uC = \{z \in \CN : |z| = 1\}$ stands for the unit circle in the complex plane $\CN$.

\paragraph*{Problem Statement}
Let $x \in \Sigma^{N}_{k}$ and let $\{ \varphi_{m} \}^{M}_{m=1}$ be a set of measurement vectors in $\CN^{N}$.
The compressive phase retrieval (CPR) problem is to reconstruct $x$ from the intensity measurements
\begin{equation}
\label{equ:measurements}
	b_{m} = \left|\left\langle x , \varphi_{m}\right\rangle\right|^{2}\,,
	\quad m = 1,\dots,M\;.
\end{equation}
Recovery will only be unique up to a unitary constant because if $x$ satisfies \eqref{equ:measurements} then also and $c x$ with $|c|=1$ will satisfy \eqref{equ:measurements}.
Consequently, we always consider \eqref{equ:measurements} as a mapping $\mathcal{A}_{\Phi} : \CN^{N}\backslash\TT \to \RN^{M}$ from the quotient space of $\CN^{N}$ modulo $\TT$ into $\RN^{M}$.


\section{General Approach}

This section proposes a general approach for compressive phase retrieval problem. Thereafter, we will give some concrete realizations applicable to Fourier optics systems such as Fig.~\ref{fig:Optics}.
We propose to split the whole recovery problem into a two step procedure: a phase retrieval step and a sparse recovery step.
More precisely, our methodology is based on the following two ingredients.
\begin{itemize}
\item[(i)]
Let $A \in \CN^{L\times N}$ such that every $x\in\Sigma^{N}_{k}$ can be recovered from the measurements $y = A x \in \CN^{L}$ as a solution of
$\min_{z\in\CN^{N}} \|z\|_{0}$ subject to $A z = y$.
\item[(ii)] Let $\Psi = \{\psi_m\}^{M}_{m=1}$ be a set of $M$ vectors in $\CN^{L}$
such that the mapping $\mathcal{A}_{\Psi} : \CN^{L}\backslash\TT \to \RN^{M}$ is injective.
\end{itemize}
Therewith, we define the measurement vectors 
\begin{equation}
\label{equ:Amatrix}
	\varphi_{m} := A^{*} \psi_{m}\in \CN^{N}\,, \quad
	m=1,\dots, M\;.
\end{equation}
By this construction of measurement vectors, one guarantees that every $k$-sparse vector in $\CN^{N}$ can be recovered from the measurements \eqref{equ:measurements}, provided the number $M$ of measurements is large enough.

\begin{theorem}
\label{thm:NumbMeasure}
If $M \geq 8 k - 4$ then there exist sets of measurement vectors $\{ \varphi_{m} \in \CN^{N} \}^{M}_{m=1}$ such that every $x \in \Sigma^{N}_{k}$ can be recovered from the quadratic measurements \eqref{equ:measurements}, up to unitary factor.
\end{theorem}

\begin{proof}
Let $x \in \Sigma^{N}_{k}$.
By the definition of the measurement vectors ${\varphi_{m}}$ in \eqref{equ:Amatrix}, we can write \eqref{equ:measurements} as
\begin{equation}
\label{equ:measurement}
	b_{m}
	= \left|\left\langle x, A^{*}\psi_{m}\right\rangle_{\CN^{N}}\right|^{2}
	= \left|\left\langle A x, \psi_{m}\right\rangle_{\CN^{L}}\right|^{2}
	= \left|\left\langle y, \psi_{m}\right\rangle_{\CN^{L}}\right|^{2}
\end{equation}
with $y = A x \in \CN^{L}$.
It is known \cite{Conca_Algebraic13} that if $M \geq 4L-4$ than there are a sets of measurement vectors $\{ \psi_{m} \in \CN^{L} \}^{M}_{m=1}$ which have property (ii).
It follows that $y\in\CN^{L}$ can be determined from the magnitude measurements $\{b_{m}\}^{M}_{m=1}$ given in \eqref{equ:measurements}, up to a unitary constant.
Moreover, if $L \geq 2 k$ then it is known \cite{CohenDahmenDeVore_09} that there exist matrices $A \in \CN^{L\times M}$ which have the property (i).
Consequently the $k$-sparse vector $x$ can be reconstructed from $y$.
\end{proof}

\begin{remark}
The previous result and the proof are similar to \cite{LiVoro_2012}.
However, our approach gives immediately explicit constructions of measurement systems as well as corresponding recovery algorithms.
In particular, there exist explicit constructions for matrices $A$ which have property (i) \cite{RauhutBook}, and there exist several known systems of vectors which have property (ii) \cite{Bodmann_StablePR2014,Fickus_VeryFewMeasurements}.
\end{remark}

The number of necessary measurements, given in Theorem~\ref{thm:NumbMeasure} is based on the known results on the minimal number of measurements necessary for the phase retrieval step and compressive sensing step.
To get stable recovery algorithms, one may need more measurements than required in Theorem~\ref{thm:NumbMeasure}.
Following the described methodology, one has to choose concrete realizations for  $A$ and $\Psi$ and different algorithms for both recovery steps. For example,
\begin{itemize}
\item Choose $\Psi$ as random vectors as in~\cite{Candes_PhaseLift} and use {\em PhaseLift} for the recovery in the phase retrieval step.
\item Pick $A$ as a random matrix and then solve the {\em basis pursuit} problem in step $2$. 
\end{itemize}




\section{CPR -- Gaussian Measurements}

In the following, we will give a concrete realization of the previously introduced methodology to sparse phase retrieval which yields a low complexity recovery algorithm.
To this end, we have to choose a matrix $A$ with property (i) and vectors $\{\psi_{m}\}$ with property (ii).
For the set $\{\psi_{m}\}$, we use the vectors proposed in \cite{PYB_SampTa14}:

\paragraph*{A set of measurement vectors}
Consider the set of $M = 4L - 4$ measurement vectors $\Psi = \{ \psi_{s,l} \}^{s=1,\dots,4}_{l=1,\dots,L-1}$ in $\CN^{L}$ given by:
\begin{equation}
\label{equ:MeasVect}
\begin{array}{rcr}
	\psi_{1,l}  =  \alpha\, e_{1} + \beta\, e_{l+1}, & \quad  &\psi_{3,l}  =  \phantom{-}\alpha\, e_{1} - \beta\, e_{l+1},\\[0.5ex]
	\psi_{2,l}  =  \beta\, e_{1} + \alpha\, e_{l+1}, & \quad  & \psi_{4,l}  =  -\beta\, e_{1} + \alpha\, e_{l+1},
\end{array}
\end{equation}
where $\{e_{l}\}^{L}_{l=1}$ is the canonical orthonormal basis in $\CN^{L}$,
and with
\begin{equation}
\label{equ:AlphaBeta}
	\alpha = \sqrt{\tfrac{1}{2}\left( 1-\tfrac{1}{\sqrt{3}} \right)}
	\quad\text{and}\quad
	\beta = \E^{-\I 5\pi/4}\sqrt{\tfrac{1}{2}\left( 1+\tfrac{1}{\sqrt{3}} \right)}.
\end{equation}
It was shown in \cite{PYB_SampTa14} that the mapping $\mathcal{A}_{\Psi} : \CN^{L}\backslash\TT \to \RN^{4L-4}$ associated with the set $\Psi$
is injective on the subspace $\mathcal{S}_{\Psi} = \{ y \in \CN^{L}/\uC\ :\ y[1]\neq 0\}$.

\begin{theorem}
\label{th:1}
Let $A \in \CN^{L\times N}$ be a Gaussian or Bernoulli random matrix, and set
\begin{equation*}
	\varphi_{s,l} := A^{*} \psi_{s,l}\in \CN^{N}\,,
	\quad s=1,\dots,4,\ l=1,\dots,L-1\;,
\end{equation*}
with the vectors $\psi_{s,l}$ defined in \eqref{equ:MeasVect}.
If $M = 4L-4 \gtrsim 8 k \ln(N/k)$, then every 
$x \in \mathcal{V} := \{ x \in \Sigma^{N}_{k} : (Ax)[1] \neq 0\}$
can be recovered (up to a global phase) from the intensity measurements
\begin{equation*}
	b_{s,l} = \left|\left\langle x , \varphi_{s,l}\right\rangle\right|^{2}\,,
	\quad s=1,\dots,4 \,,
		\quad l=1,\dots,L-1.
\end{equation*}
with high probability.
\end{theorem}
\begin{proof}
As in Theorem~\ref{thm:NumbMeasure}, by the definition of the measurement vectors, we have $b_{s,l} = | \left\langle y,\psi_{s,l} \right\rangle_{\CN^{L}} |^{2}$ with $y := A x \in \CN^{L}$.
Since $\mathcal{A}_{\Psi} : \CN^{L}\backslash\TT \to \RN^{4L-4}$ is injective on $\mathcal{S}_{\Psi}$ and $y[1]\neq 0$, $y$ can be reconstructed from $\{b_{s,l}\}$, up to a global phase.
By the assumption of the theorem $L > 2 k \ln(N/k)$. Therefore, standard CS theory guarantees \cite{RauhutBook} that $A$ has the null space property (with high probability) and so $x \in \Sigma^{N}_{k}$ can be recovered from the linear measurements $y = A x \in \CN^{L}$.
\end{proof}

\begin{remark}
The estimate for the necessary number of measurements $M$ and the notion "with high probability" can be made precise using well known results from CS, see, e.g., \cite{RauhutBook}.
\end{remark}

\begin{remark}
Note also, that our construction provides a natural recovery algorithm.
In the first step, $y = Ax$ is determined from the measurements $\{b_{s,l}\}$ using the algorithm proposed in \cite{PYB_SampTa14}.
Afterwards the $k$-sparse vector $x$ can be determined from $y$ by any algorithm known for sparse signal recovery.
For concreteness, we assume that \emph{basis pursuit} is used, i.e., $x \in \Sigma^{N}_{k}$ is the unique solution of the 
following convex minimization problem:
\begin{equation}
\label{equ:l1F}
	\min_{z\in\CN^{N}} \|z\|_{\ell_{1}}\quad
	\text{subject to}\quad
	A z = y\;.
\end{equation}
\end{remark}

\section{CPR -- Fourier Measurements}
\label{sec:Fourier}

In many applications, the matrix $A$ in \eqref{equ:Amatrix} cannot be determined arbitrarily. Here, we adapt our idea from last section to the setup in Fig.~\ref{fig:Optics}. In this setting, the object of interest is illuminated and the resulting diffraction pattern $x[n]$ is modulated by suitable masks with transmittance functions $p_{s}[n]$, such that $x\odot p_{s}$ is the resulting signal after each mask. Subsequently, the lens transforms the modulated signal into the frequency domain. As we are interested in recovering spatially sparse signals $x\in \Sigma^{N}_{k}$, we exploit this sparsity and take only random frequency measurements $M < N$, where $M$ is determined by the compressive sensing theory. 

In particular, we propose to use four masks with the following transmittance functions
\begin{equation}
\label{equ:Masks}
\begin{array}{rcr}
	p_{s}[n]  =  a_{s}\, \delta[n] + b_{s}\;,\quad n = 1,\dots,N\,
\end{array}
\end{equation}
where $\delta[n]$ stands for the delta function defined by $\delta[1] = 1$ and $\delta[n] = 0$ for $n\neq 1$, and with the constants
\begin{equation*}
\begin{array}{llll}
	a_{1} = \alpha\,,\ & a_{2} = \beta\,,\ & a_{3} = \phantom{-}\alpha\,,\ & a_{4} = -\beta\,,\\
	b_{1} = \beta\,,\ & b_{2} = \alpha\,,\ & b_{3} = -\beta\,,\ & b_{4} = \phantom{-}\alpha\,,
\end{array}
\end{equation*}
and where $\alpha$ and $\beta$ are defined as in \eqref{equ:AlphaBeta}.
Based on these masks, we can prove the following recovery result.


\begin{theorem}
Consider the measurement setup of Fig.~\ref{fig:Optics} with the four masks $p_s \in \CN^{N}$ as defined in \eqref{equ:Masks}.
Let $\mathcal{L} \subset \{1,\dots,N\}$ be a set of randomly chosen sampling points in the Fourier domain.
If
\begin{equation*}
	L = |\mathcal{L}| > C\, k \log(N), 
\end{equation*}
with an appropriated constant $C$, then every $x \in \Sigma^{N}_{k}$ with $x[1] \neq 0$ can be recovered (up to an global unitary phase) from the $4 L$ intensity measurements
\begin{equation*}
	b_{s,l} = |\mathcal{F}(x \odot p_s)[l]|^{2}\,,
	\quad s=1,\dots,4\,,\ l\in \mathcal{L},
\end{equation*}
with high probability.
\end{theorem}

\begin{remark}
So any $k$-sparse vectors in $\CN^{N}$ can be recovered from $M \gtrsim 4 C k \log(N)$ intensity measurements. 
The constant $C$ and the statement ``with high probability'' can be made more precise using result on CS with partial Fourier measurements \cite{RUP}.
We will show in Sec.~\ref{sec:Simulation}, by means of numerical simulations, that we need approximately $M \gtrsim 4 k \log_2(N/k)$ measurements for recovery with high probability.
\end{remark}

\begin{remark}
Note also, that we have again a very mild restriction on the signal space, namely that the first signal entry $x[1]$ must not vanish.
The restriction is necessary to allow phase retrieval in the first recovery step \cite{PYB_SampTa14}.
\end{remark}

\begin{proof}
Direct calculation shows that
\begin{equation*}
	\widehat{y}_{s} = \mathcal{F}(x \odot p_s) = D_{s} \widetilde{x}
\end{equation*}
where $\widehat{y}_{s} \in \CN^{L}$ is the vector of the Fourier transform of $x \odot p_s$ sampled at the points $\mathcal{L}$, the vector
\begin{equation*}
	\widetilde{x}
	= ( \frac{x[1]}{\sqrt{N}} , \widehat{x}^{\T} )^{\T}
	= ( \frac{x[1]}{\sqrt{N}} , \widehat{x}[l_{1}] , \dots , \widehat{x}[l_{L}])^{\T} \in \CN^{L+1}
\end{equation*}
contains $\frac{x[1]}{\sqrt{N}}$ at the first position and the Fourier transform $\mathcal{F} x$, sampled at the set $\mathcal{L} = \{l_{1},\dots,l_{L}\}$ at the other positions, and $D_{s}\in\CN^{L \times (L+1)}$ are matrices of the form 
\begin{equation*}
D_s = 
	\begin{pmatrix}
	   \overline{a_{s}}&\overline{b_{s}}& 0 &\cdots &\cdots &0\\
	   \overline{a_{s}}&0&\overline{b_{s}}&0&\cdots &0\\
	   \vdots&\vdots&\ddots&\ddots&\ddots&\vdots \\
	   \overline{a_{s}}& 0&\cdots&0&\overline{b_{s}}&0\\
	   \overline{a_{s}}&0 &\cdots & \cdots& 0&\overline{b_{s}}\;.   
	\end{pmatrix}
\end{equation*}
From the simple structure of $D_{s}$, a direct calculation shows that the measurements can be written as
\begin{equation}
\label{equ:measure2}
	b_{s,l} = |\widehat{y}_{s}[l]|^{2} = |\left\langle \widetilde{x} , \psi_{s,l} \right\rangle|\,,
	\quad s=1,\dots,4,\ l=1,\dots,L
\end{equation}
where the set $\Psi = \{ \psi_{s,l} \}$ of $\CN^{L+1}$-vectors is defined as in \eqref{equ:MeasVect}.
Again, we use that the corresponding mapping $\mathcal{A}_{\Psi} : \CN^{L+1}\backslash\TT \to \RN^{4 L}$ is injective.
Consequently, $\widetilde{x} \in \CN^{L+1}$ can be recovered from the measurements \eqref{equ:measure2}, up to a constant phase factor.
Discarding the first entry of $\widetilde{x}$, we obtain in particular $\widehat{x}$ which can be written as
\begin{equation*}
	\widehat{x} = F_{\mathcal{L}} x
\end{equation*}
where $F_{\mathcal{L}} \in \CN^{L \times N}$ stands for the partial DFT matrix with the $L$ rows, indexed by the random set $\mathcal{L}$, of the $N\times N$ DFT matrix $F$.
Since $|\mathcal{L}| > C k \log(N)$, it is known \cite{RUP} that the $k$-sparse vector $x$ can be recovered from $\widehat{x}$ (with high probability) using \eqref{equ:l1F} with $A = F_{\mathcal{L}}$.
\end{proof}
\section{Numerical Simulations}
\label{sec:Simulation}

In this section, we present numerical simulations to support and discuss our theoretical results.
Thereby, we will concentrate us on the setup based on Fourier measurements discussed in Sec.~\ref{sec:Fourier}.
As described, the overall recovery algorithm is based on the algebraic algorithm for PR described in \cite{PYB_SampTa14}, followed by Basis Pursuit \eqref{equ:l1F} which was implemented using SPGL1~\cite{spgl1:2007}.

\begin{figure}[ht]
	\centering
	\raisebox{0.52in}{\rotatebox{90}{empirical success rate}}\,
     \includegraphics[width=2.7in]{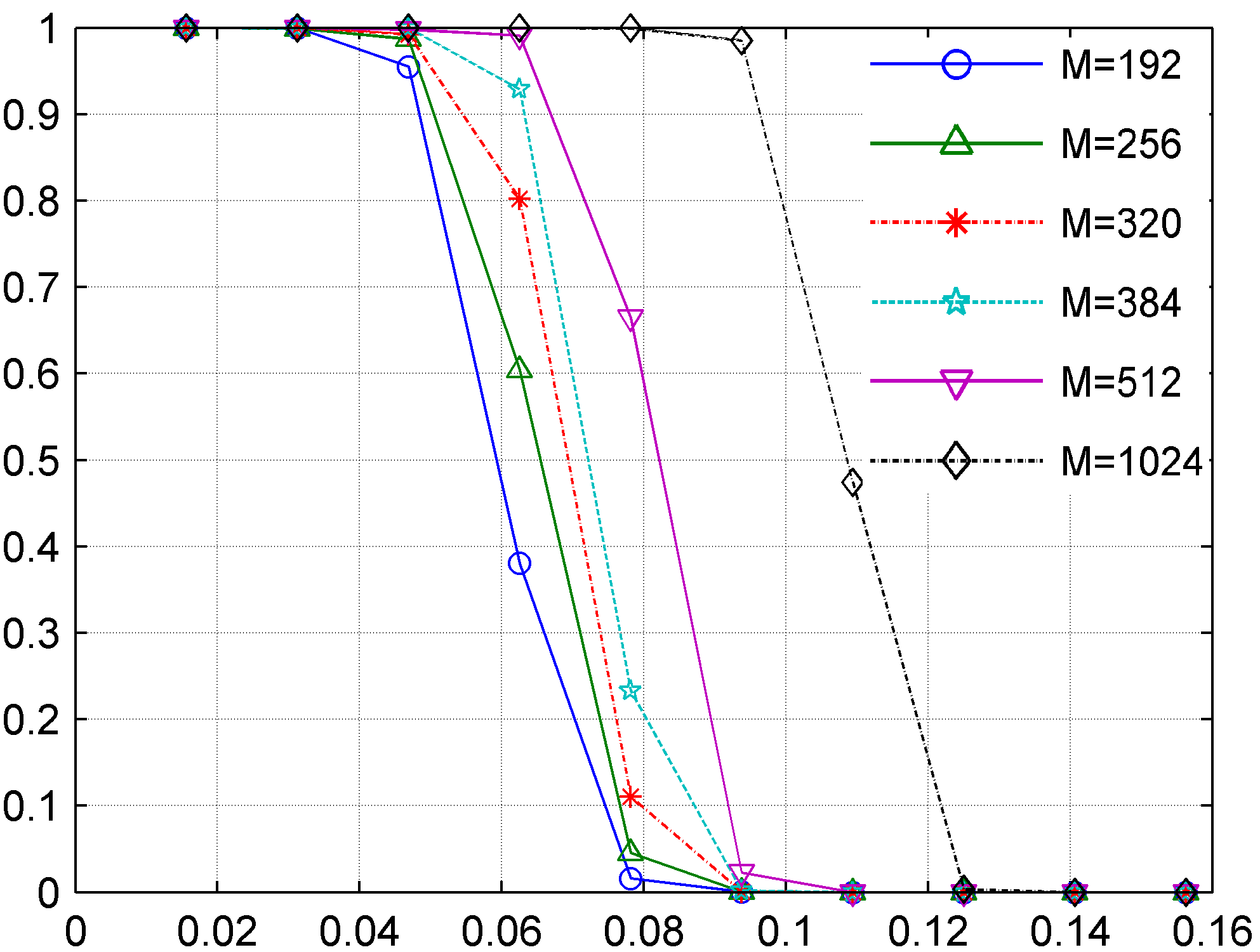}\\
		{ $k/M$}
	\caption{Empirical success rate versus $k/M$ for different number of measurements $M$ and signal dimension $N=512$.}
	\label{fig1}
\end{figure}

Throughout, we used signals $x$ of dimension $N = 512$.
Our test signals are $k$-sparse with a uniformly randomly chosen support with independent, but equally distributed complex Gaussian variables with variance 1.
Measurements are assumed to be disturbed by additive white Gaussian noise
\begin{equation*}
	b_{s,l} = \left| \mathcal{F}(x\odot p_s)[l] + \nu_{s,l} \right|^2\;,
	\quad
	\begin{array}{l}
		s = 1,\dots,4\\
		l = 1,\dots,|\mathcal{L}|\;,	
	\end{array}
\end{equation*}
where $\nu_{s,l} \sim \mathcal{N}(0,\sigma_{\nu}^{2})$ are independent, normally distributed complex random variables with variance $\sigma^{2}_{\nu}$. The signal-to-noise ratio (SNR) is defined as $SNR = \|\mathcal{F}(x\odot p_{s})\|^{2}/E[\|\nu\|^{2}]$.
After we recovered $x$ from the noisy measurements, we determined the relative mean squared error $MSE = \|x - \widetilde{x}\|^{2}_{2}/\|x\|^{2}_{2}$ where $\widetilde{x}$ stands for the estimated signal with the corrected phase.

First, we examined the noiseless case at $SNR = 100$dB.
We performed $2000$ simulations to investigate empirically the number of measurements $M$ which are necessary to recover a $k$-sparse signal.
The results are shown in Fig.~\ref{fig1}. It shows that we need approximately $M \approx 20 k$ measurements for small sparsity values $k\approx 10$, and $M \approx 10 k$ for sparsity values of about $k=100$.

\begin{figure}[ht]
	\centering
	\raisebox{0.95in}{\rotatebox{90}{ $M$}}\,
     \includegraphics[width=2.7in]{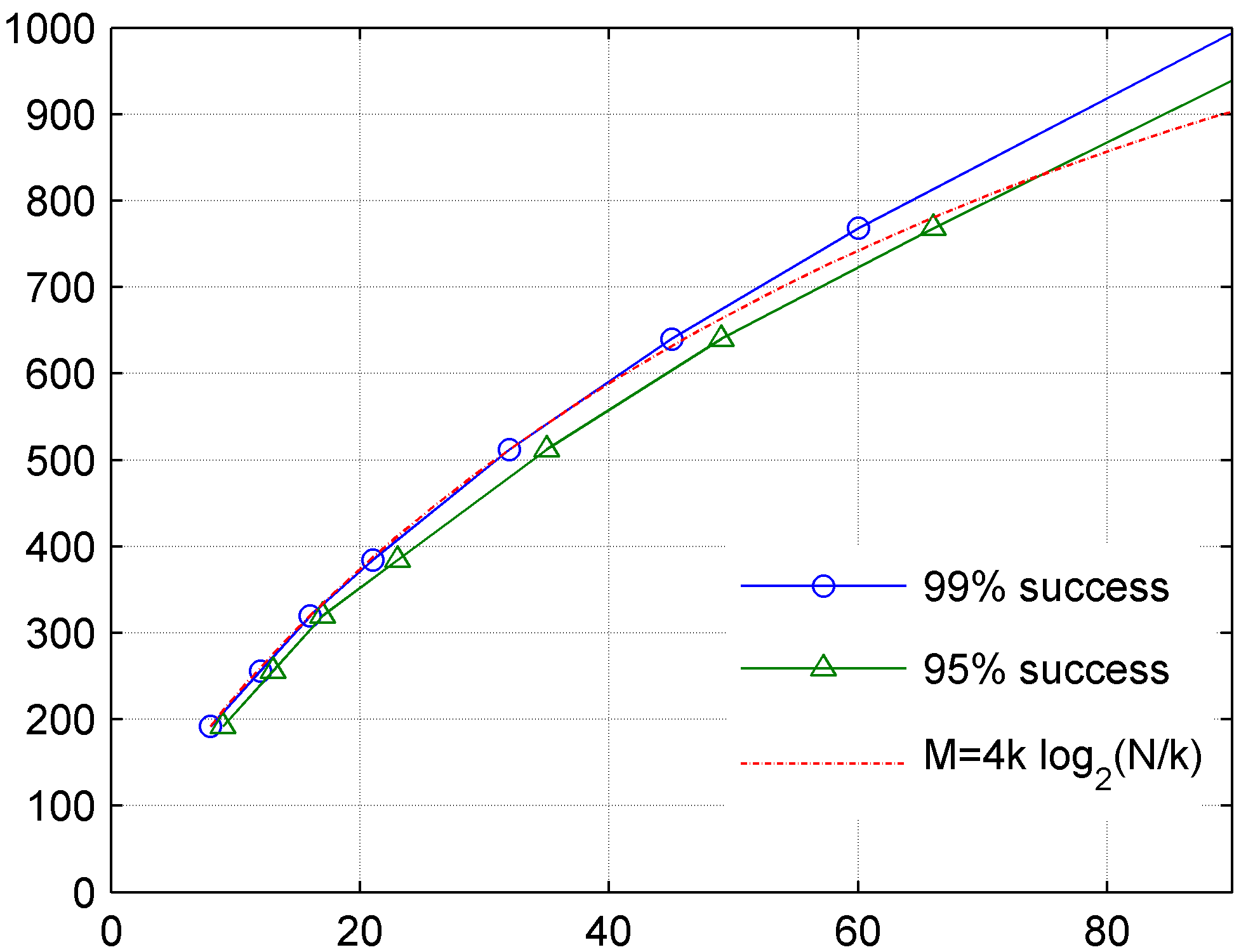}\\
		{ $k$}\\
	\caption{Number of measurements $M$ necessary to recover a $k$-sparse signal $x \in \Sigma^{N}_{k}$ of dimension N=512.}
	\label{fig2}
\end{figure}

To investigate the relation between the necessary number of measurements $M$ and the sparsity $k$ further, Fig.~\ref{fig2} plots $M$ versus $k$ for success rates of $99\%$ and $95\%$, respectively.
Thereby, we regarded a reconstruction as successful whenever the MSE was less than $10^{-5}$.
For small values of $k$, the graphs are well approximated by the relation $M \approx 4 k\log_2 (N/k)$, which is also shown.

\begin{figure}[ht]
	\centering
	\raisebox{.45in}{\rotatebox{90}{ normalized MSE (dB)}}
     \includegraphics[width=2.7in]{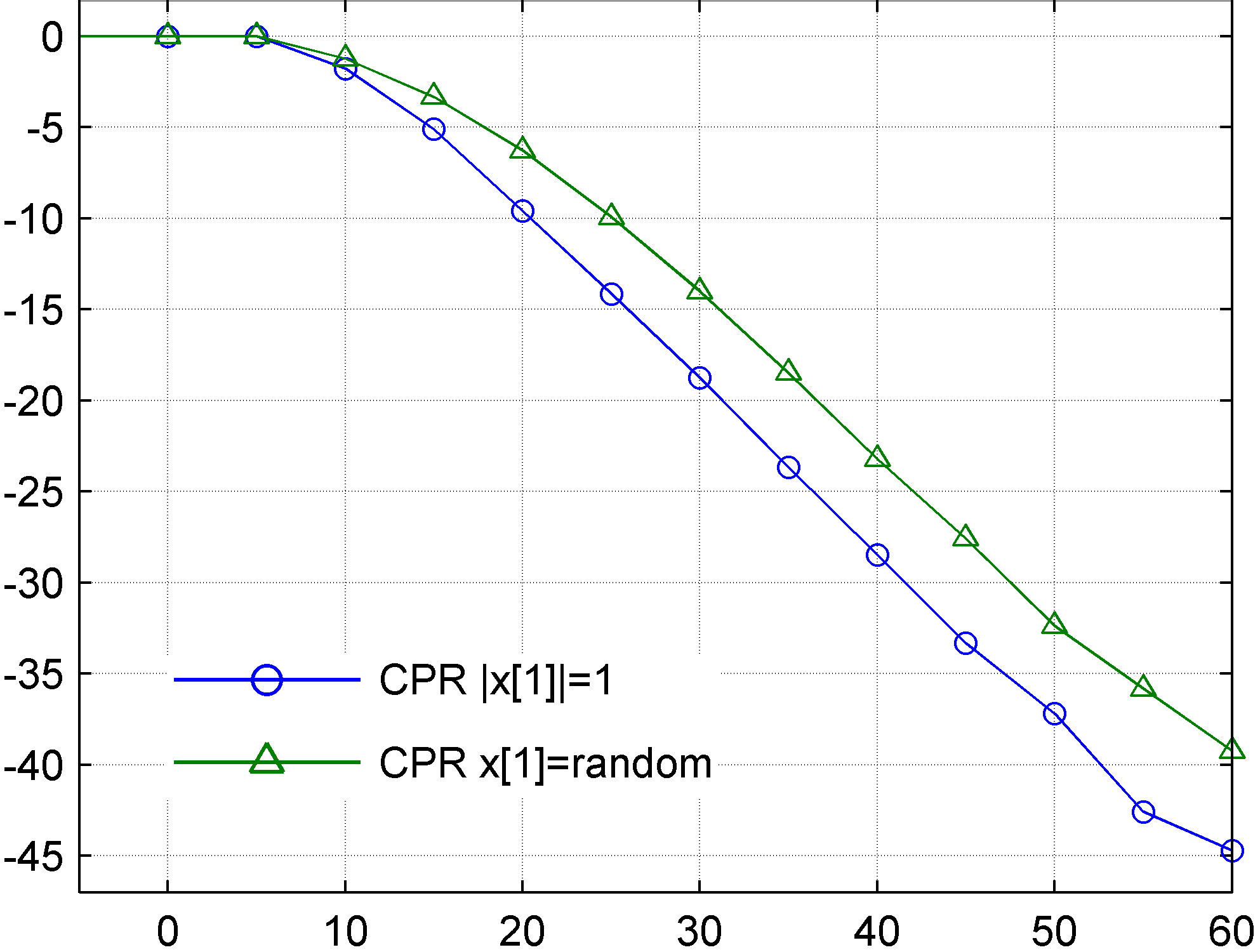}\\
		{ SNR (dB)}
	\caption{Normalized MSE versus SNR for the recovery algorithm proposed in Sec.~\ref{sec:Fourier} for sparse signals $x \in \Sigma^{N}_{k}$ with $N=512$, $k= 12$ and with $M = 256$ intensity measurements.}
	\label{fig3}
\end{figure}

Finally, we studied the stability behavior of the recovery algorithm under additive noise.
Simulation results for signal dimension $N=512$, sparsity $k=12$ and $M = 256 $ measurements are shown in Fig.~\ref{fig3}.
The simulation results were averaged over $10^3$ trials.
In the simulations, we distinguished also between the situation, where the amplitude of the first signal entry was fixed $|x[1]|=1$, and where it was chosen randomly, respectively.
We see that the overall algorithm is stable under additive noise. The simulations show that the reconstruction error $\|x - \widetilde{x}\|^{2}$ is approximately proportional to the norm $\|\nu\|^{2}$ of the additive noise. One obtains a slightly better performance, if the amplitude of the first signal component $x[1]$ is fixed.

\section{Discussion -- Outlook}

The approach of sparse phase retrieval, presented in this paper, is based on a two step recovery procedure: 1) a phase retrieval step, 2) a sparse signal recovery step.
For the first step, we proposed to apply the algebraic algorithm proposed in \cite{PYB_SampTa14}, for the sparse recovery, we proposed to use basis pursuit (BP).
The advantage of this composition, is that the algebraic phase retrieval algorithm has a very low complexity, which only scales linearly with the number $M$ of measurements.
So the overall complexity is mainly determined by the $\ell_{1}$ minimization \eqref{equ:l1F} of basis pursuit which only operates in the dimension $L=M/4$.
Moreover, since both separate algorithms are stable, also the overall sparse phase retrieval is stable. 
The derivations of concrete error bounds is left as a future work. It will be based on the known stability analysis for BP and the \cite{PYB_SampTa14}. 
 
In particular, it was shown that the proposed methodology can also be used to design deterministic masks for practical setups as in Fig.\ref{fig:Optics}, which are based on Fourier measurements.
Our analysis and simulations showed no degradation of the performance for such Fourier measurements, compared to Gaussian measurements, as observed in \cite{prGAMP}.

By the two-staged nature of our recovery methodology, other methods for phase retrieval may be applied as well.
For example, \cite{phasecode} recently proposed masks similar to \eqref{equ:Masks} for phase retrieval, but where only $3$ instead of $4$ masks are needed.
Applying these masks instead of \eqref{equ:Masks} would reduce the overall number of measurements, and it would be interesting to compare the stability behavior with those masks.

\bibliographystyle{IEEEtran}
\bibliography{pub}

\end{document}